\documentclass[12 pt]{article}
\counterwithin{equation}{section}
\usepackage[all]{xy}
\usepackage{amsmath,amssymb,amsfonts,latexsym,float,graphics,epsfig}

\usepackage{setspace}
\usepackage{appendix}
\usepackage{xcolor}
\usepackage[top=2cm, bottom=2cm, left=2cm, right=2cm]{geometry}
\usepackage[colorlinks]{hyperref}
\usepackage[all]{xy}           
\usepackage{amsmath,amsfonts,amssymb,amsthm,textcomp}
\usepackage{eucal}
\usepackage{epsfig}

\def\v #1{\vert #1\vert}             

\def\m #1 #2{(-1)^{{\v #1} {\v #2}}} 

\theoremstyle{plain}
\newtheorem{theorem}{Theorem}

\theoremstyle{definition}

\def\<#1>{\langle#1\rangle}

\usepackage{epsfig}
\usepackage{amsmath}
\usepackage{epic}
\usepackage{amssymb}
\usepackage{fancybox}
\usepackage{wrapfig}
\usepackage[font=footnotesize,labelsep=period]{caption}

\usepackage{float}

\usepackage{booktabs}

\usepackage[all]{xy}

\usepackage{fancyhdr}

\usepackage{fancyhdr}

\newcommand{\Reeb}{\mathcal{R}}
\newcommand{\dd}{\mathrm{d}}
\newcommand{\dv}[2]{\frac{\dd #1}{\dd #2}}

\newcommand*{\contr}[1]{i_{#1}}

\usepackage{tikz-cd}

\begin{document}

\centerline{\Large \bf A Discrete Hamilton--Jacobi Theory}
\medskip
\centerline{\Large \bf for Contact Hamiltonian Dynamics}\vskip 0.25cm

\medskip
\medskip

\begin{center}
O\u{g}ul Esen\footnote{E-mail: 
\href{mailto:oesen@gtu.edu.tr}{oesen@gtu.edu.tr}}\\
Department of Mathematics, \\ Gebze Technical University, 41400 Gebze,
Kocaeli, Turkey.

\bigskip

Cristina Sard\'on\footnote{E-mail: \href{mailto:mariacristina.sardon@upm.es}{mariacristina.sardon@upm.es}}
\\ Department of Applied Mathematics 
\\ Universidad Polit\'ecnica de Madrid 
\\ C/ Jos\'e Guti\'errez Abascal, 2, 28006, Madrid. Spain.

\bigskip 

Marcin Zajac\footnote{E-mail: \href{marcin.zajac@fuw.edu.pl}{marcin.zajac@fuw.edu.pl}}
\\
Department of Mathematical Methods in Physics,\\
Faculty of Physics. University of Warsaw,\\
ul. Pasteura 5, 02-093 Warsaw, Poland.

\end{center}


\begin{abstract} 
In this paper, we propose a discrete Hamilton--Jacobi theory for (discrete) Hamiltonian dynamics defined on a (discrete) contact manifold. To this end, we first provide a novel geometric Hamilton--Jacobi theory for continuous contact Hamiltonian dynamics. Then, rooting on the discrete contact Lagrangian formulation, we obtain the discrete equations for Hamiltonian dynamics by the discrete Legendre transformation. Based on the discrete contact Hamilton equation, we construct a discrete Hamilton--Jacobi equation for contact Hamiltonian dynamics. We show how the discrete Hamilton--Jacobi equation is related to the continuous Hamilton--Jacobi theory presented in this work. Then, we propose geometric foundations of the discrete Hamilton--Jacobi equations on contact manifolds in terms of discrete contact flows. At the end of the paper we provide a numerical example to test the theory. 
\\

\end{abstract}

\setlength{\parskip}{4mm}

\onehalfspacing
\section{Introduction}

This work lies in the intersection among the geometric Hamilton-Jacobi (abbreviated as HJ) theory, discrete dynamics and contact geometry. A HJ theory for discrete Hamiltonian dynamics on contact manifolds was missing in the literature, this is why our aim in this work is to fill this gap by introducing a discrete HJ equation for discrete contact Hamiltonian dynamics. Additionally, we shall examine the geometric foundations of the discrete HJ equation. Accordingly, we shall propose a geometric discrete HJ theorem in the contact framework.  

The Hamilton-Jacobi equation was first given for classical and continuous Hamiltonian dynamics on symplectic manifolds \cite{Arnold-book,Goldstein-book}. More recently, a geometrization of the HJ equation was established in \cite{carinena2006geometric}. Since this work, various applications and generalizations of the geometric HJ theory have been exhibited. 
We refer to two recent surveys \cite{esen2022reviewing,roman2021overview} (and references therein) for a contemporary view of the current geometric HJ theory. 

On the other hand, in recent years there has been a growing trend in providing proper discrete analogs of continuous differential equations and designing numerical methods adapted to the type of discretization pursued, as well as the type of differential equations that are dealt with. Numerical methods have shown their utility in solving equations that cannot be managed analytically \cite{At,iserles1}. In the branch of geometric mechanics, we deal with a plethora of geometric structures that provide different underlying geometric properties to dynamical systems. This is why when one discretizes a dynamical system, one has to make sure that the discretization is compatible with the geometric structure and that we are applying specific methods that preserve the geometric structure. These specific methods are known as geometric integrators \cite{BlanesCasas,Hairer}. For example, in classical mechanics we propose numerical methods that preserve the symplectic structure when we work on a phase space \cite{Candy,Yo}, others methods are energy-preserving numerical methods \cite{Quispel}, momentum-preserving methods \cite{Labudde}, etc.

Contact geometry is a popular theme in the recent literature \cite{BrCrTa17,LeonLainz2,deLeon2019a,de2020review,LeonLainz3}, being widely used to describe mechanical dissipative systems, dissipative field theories and generalizations of the Hamilton principle \cite{Herglotz,herg2}. Some of the main uses of contact geometry and its main characteristics can be consulted in \cite{Br17}. The dissipative character of the formalism provides an important geometric foundation for irreversible dynamics, especially thermodynamics. Here is an incomplete list of works related to contact mechanics and its role in thermodynamics \cite{Bravetti-thermo,gaset2020new,Goto,Grmela-Contact,Mrugala,Rajeev-contact}.

To propose a discrete Hamilton--Jacobi theory on a contact manifold as we shall, one needs to review first the discrete formulation of mechanics on the Lagrangian side \cite{MarsWest}. This leads to the discretization of Lagrangian and Hamiltonian systems, as well as the variational principles for dynamical systems and principles of critical action on both the tangent and cotangent bundle \cite{guibout2004discrete,Marsden-LoM}. 
Such discretizations led to discretized versions of Noether's theorem, Legendre transformations, infinitesimal symmetries, etc.
The discretization of the Hamiltonian formulation gave rise to optimal control problems by developing a discrete maximum principle that yields discrete necessary conditions for optimality. Furthermore, discrete Hamiltonian theories have been particularly useful in distributed network optimization
and derivation of variational integrators \cite{Lall}. 
The geometry of the space is also a key point to performing better discretizations. For this matter, it is important to rely on symmetries and invariants of the geometric space \cite{budd}.
 In this work we preserve the contact structure under discretization, minimizing the error in the approximation. Some very recent works addressing discrete Lagrangian and Hamiltonian dynamics on contact settings are \cite{Bravetti2,Bravetti1,Simoes2020b,Simoes21,Vermeeren}. In these works, one can see the discrete generalized Lagrangian (Herglotz) dynamics on the extended tangent bundle as well as the discrete Hamiltonian dynamics on a contact manifold. Since these works are fundamental for the present study, we shall give a quick review of the theories in the main body of the paper.

We refer to \cite{OhsawaBlochLeok} for the discrete HJ equation on a symplectic manifold. In \cite{OhsawaBlochLeok}, the HJ equation is derived by employing discrete symplectic flows. In the present work, in similar fashion, we carry this discussion to contact geometry. The role of the discrete symplectic flows will be played by discrete contact flows analogously.

More recently, in  \cite{LeonSar3}, the geometrization of the discrete HJ equation in \cite{OhsawaBlochLeok} has been established in the symplectic category. In this regard, it is possible to consider the present work as a continuation or an extension of the works \cite{LeonSar3,OhsawaBlochLeok} to contact geometry and discrete contact Hamiltonian dynamics. In the present paper, we both derive the discrete HJ equation on contact manifolds and then proceed with its geometrization. 

The main body of this work contains three sections. In the upcoming  Section \ref{CCD-Section}, we shall first review the fundamental principles of contact manifolds as well as Lagrangian and Hamiltonian dynamics on contact manifolds, in order to fix the notation. In Subsection \ref{contact:hamiltonian}, we shall provide a HJ theorem for  contact Hamiltonian dynamics.  
We shall start Section \ref{DCD-Section} by presenting discrete Lagrangian and Hamiltonian contact dynamics. In subsection \ref{discrete-HJ-Section}, we shall present our main result which is a Hamilton--Jacobi equation for discrete contact Hamiltonian dynamics. We shall prove this result in terms of discrete contact flows. 
In Subsection \ref{gHJ-Section}, a geometrization of the HJ is provided. 
In Section \ref{Application}, we conclude this work by proposing a numerical example applied to the well-known parachute equation in contact dynamics.
To avoid mathematical conflict and without loss of generalization, we assume all objects to be smooth and globally defined unless stated otherwise. Manifolds are connected and differentiable.

\section{Fundamentals of Continuous Contact Dynamics}\label{CCD-Section}

In this section, we first recall briefly the main definitions and results of the theory of Lagrangian and Hamiltonian dynamics on contact manifolds following \cite{BrCrTa17,deLeon2019a,de2020review}. Later we introduce a geometric Hamilton-Jacobi theory for Hamiltonian dynamics on contact manifolds. This novel theorem will be the continuous version of the discrete HJ theorem that will be presented in the upcoming section. 

\subsection{Contact Manifolds}\label{Sec-Cont-Man}

We call a contact manifold a pair $(M,\eta)$, where $M$ is an odd-dimensional manifold, say $(2n+1)$-dimensional, with a contact form $\eta$, i.e., a one-form on $M$ such that $\eta \wedge \dd \eta^n\neq 0$ is a volume form. This type of  manifold has a distinguished vector field, the Reeb vector field $\Reeb$, which is the unique vector field that satisfies the two following identities.
\begin{equation}
	\contr{\Reeb} \dd \eta = 0, \qquad \eta(\Reeb)=1.
\end{equation}
There exists a Darboux coordinate system $(\mathbf{q}, \mathbf{p},s)=(q^i, p_i, s)$ on $M$ (with $i$ ranging from 1 to the dimension of $M$) such that the contact one-form reads
\begin{equation}\label{can-contat-form}
	 \eta = \dd s - \mathbf{p} \cdot  \dd \mathbf{q}.
\end{equation}
In these coordinates, we have $\Reeb = \nabla_s = \partial/ \partial s$. This local observation provides an example of a contact manifold as the extended cotangent bundle 
\begin{equation}\label{eq:cotangent_contact_structure}
\big(T^*Q\times \mathbb{R}, \eta  = \dd s - \theta_Q\big) 
\end{equation}
where $\theta_Q$ is the  pullback of the tautological one-form of $T^*Q$.

\textbf{Musical Morphisms.} For a contact manifold $(M, \eta)$, there is a musical isomorphism 
\begin{equation}\label{flat-map}
\flat:TM\longrightarrow T^*M,\qquad v\mapsto \iota_v \dd \eta+\eta(v)\eta.
\end{equation} 
This mapping takes the Reeb field $\mathcal{R}$ to the contact one-form $\eta$. We denote the inverse of this mapping by $\sharp$. Referring to this, we define a bivector field $\Lambda$ on  $M$ as
\begin{equation}\label{Lambda}
\Lambda(\alpha,\beta)=-\dd \eta(\sharp\alpha, \sharp \beta). 
\end{equation}
Then referring to $\Lambda$ we introduce the following musical mapping 
\begin{equation}\label{Sharp-Delta}
\sharp_\Lambda: T^*M\longrightarrow TM, \qquad  \alpha\mapsto \Lambda(\alpha,\bullet)= \sharp \alpha - \alpha(\mathcal{R})  \mathcal{R}. \end{equation}
Evidently, the mapping $\sharp_\Lambda$ fails to be an isomorphism. Notice that the kernel is spanned by the contact one-form $\eta$. 

\textbf{Legendrian  Submanifolds.}
Let $(M,\eta)$ be a contact manifold. Recall the associated bivector field $\Lambda$ defined in \eqref{Lambda}. Consider a linear subbundle $ \Xi$ of the tangent bundle $TM$ (that is, a distribution on $M$). We define the contact complement of $\Xi$ as
\begin{equation}
\Xi^\perp : = \sharp_ \Lambda(\Xi^o),
\end{equation}
 where the sharp map on the right hand side is the one in   \eqref{Sharp-Delta} and   $\Xi^o$ is the annihilator of  $\Xi$.  We say that $N$ is Legendrian if $TN^{\perp }= TN$. 

Consider the contact manifold $T^*Q\times \mathbb{R}$ in \eqref{eq:cotangent_contact_structure} and let $W$ be a real valued function on the base manifold $Q$. Its first prolongation is 
\begin{equation}\label{j1F}
\rm{J}^1  W:Q\longrightarrow  T^*Q\times \mathbb{R},\qquad \mathbf{q}\mapsto \big(\mathbf{q}, W_\mathbf{q},W(\mathbf{q})\big).
\end{equation}
The image space of the first prolongation $\rm{J}^1   W$ is  a Legendrian  submanifold of $T^*Q\times \mathbb{R}$. The converse of this assertion is also true, that is, if the image space of a section   of $T^*Q\times \mathbb{R}\mapsto Q$ is a Legendrian submanifold then it is  the first prolongation of a function $W$. 
 
 \textbf{Contact Diffeomorphisms.} Let $(M,\eta)$ and $(\widehat{M},\widehat{\eta})$ be two contact manifolds. A diffeomorphism $\varphi$ from $M$ to $\widehat{M}$ is said to be a contact diffeomorphism (or contactomorphism) if it preserves the contact structures, that is, $T\varphi(\ker \eta)=\ker\widehat{\eta}$. In terms of contact forms, a contact diffeomorphism $\varphi$ is the one satisfying 
   \begin{equation}\label{Cont-Dif}
        \varphi^*\widehat{\eta} = \mu \eta.
    \end{equation}
 where $\mu $ is a conformal factor. If we additionally impose that the conformal factor $\mu$ in definition \eqref{Cont-Dif} has to be equal to one, we arrive at the conservation of the contact form. 
We call such a mapping a strict contact diffeomorphism (or quantomorphism).

Consider two contact manifolds $(M,\eta)$ and $(\widehat{M},\widehat{\eta})$. A contact product is the product manifold $\widehat{M} \times M  \times \mathbb{R}$ with a contact one-form $\tau\widehat{\eta}\ominus \eta$ where $\tau$ is a global coordinate on $\mathbb{R}$ \cite{banyaga97}. It is possible to validate that the graph of a contact diffeomorphism $\varphi$ is a Legendrian submanifold of the contact product $\widehat{M} \times M  \times \mathbb{R}$. 

\textbf{Generating Functions For Legendrian Submanifolds.}
In particular, consider two same dimensional extended cotangent bundles $T^*Q\times \mathbb{R}$ and $T^*\widehat{Q}\times \widehat{\mathbb{R}}$ equipped with  Darboux' coordinates $(\mathbf{q},\mathbf{p},s)$ and $(\widehat{\mathbf{q}},\widehat{\mathbf{p}},\widehat{s})$, respectively. Then, we determine the contact product $\big(T^*\widehat{Q}\times \widehat{\mathbb{R}} \big) \times \big( T^*Q\times \mathbb{R} \big) \times  \mathbb{R} $ of these two contact manifolds and consider it as a fiber bundle over the product manifold $\widehat{Q}\times Q \times \mathbb{R}$ given by
 \begin{equation}\label{cont-prod-d}
\big(T^*\widehat{Q}\times \widehat{\mathbb{R}} \big) \times \big( T^*Q\times \mathbb{R} \big) \times  \mathbb{R} \longrightarrow  \widehat{Q}\times Q \times \mathbb{R}, \qquad (\widehat{\mathbf{q}},\widehat{\mathbf{p}},\widehat{s};\mathbf{q},\mathbf{p},s;\tau)\mapsto (\widehat{\mathbf{q}},\mathbf{q},\widehat{s}),
\end{equation}
where $\tau$ is the coordinate on the extended $\mathbb{R}$. In the contact product manifold, the product contact one-form is defined to be 
 \begin{equation}\label{cont-prod-eta}
 \tau \widehat{\eta} - \eta=\tau (\dd\widehat{s}-\widehat{\mathbf{p}}\cdot \dd\widehat{ \mathbf{q}}) - (\dd s-\mathbf{p}\cdot \dd \mathbf{q}).
     \end{equation}
Determining the isomorphism 
 \begin{equation}
 \big(T^*\widehat{Q}\times \widehat{\mathbb{R}} \big) \times \big( T^*Q\times \mathbb{R} \big) \times  \mathbb{R} \simeq T^*(\widehat{Q}\times Q \times \mathbb{R})\times \mathbb{R}, 
   \end{equation}
we provide Darboux coordinates for the contact product manifold as 
      \begin{equation}\label{coord-concttc}
      (\widehat{ \mathbf{q}},\mathbf{q},\widehat{s};\widehat{\boldsymbol{\pi}},\boldsymbol{\pi},\widehat{\pi};z)
      =
      (\widehat{ \mathbf{q}}, \mathbf{q},\widehat{s};\tau\widehat{\mathbf{p}}, -\mathbf{p}, -\tau;-s ).
        \end{equation}
This enables us to recast the contact one-form  \eqref{cont-prod-eta} in the canonical form given in \eqref{can-contat-form}. So that the one-form \eqref{cont-prod-eta} turns out to be 
\begin{equation}
\tau \widehat{\eta} - \eta= \dd z - \widehat{\boldsymbol{\pi}}\cdot \dd \widehat{ \mathbf{q}} -
\boldsymbol{\pi} \cdot \dd \mathbf{q} -
\widehat{\pi} \dd \widehat{s}.
\end{equation}
On the base manifold $\widehat{Q}\times Q \times \mathbb{R}$, we define a smooth function in form
 \begin{equation}
W(\widehat{\mathbf{q}},\mathbf{q},\widehat{s})=\widehat{s}+S(\widehat{\mathbf{q}},\mathbf{q},\widehat{s}).
 \end{equation}
for a function $S$. The theory manifests that the image space of the first jet $\rm{J}^1  W$ is a Legendrian submanifold of the contact product manifold \eqref{cont-prod-d}.  The contact diffeomorphism determined by the Legendrian submanifold $im(\rm{J}^1  W)$ is computed to be
 \begin{equation}
 T^*Q\times \mathbb{R} \longrightarrow T^*\widehat{Q}\times \widehat{\mathbb{R}},\qquad (\mathbf{q},\mathbf{p},s)\mapsto (\widehat{\mathbf{q}},\widehat{\mathbf{p}},\widehat{s})
    \end{equation}
where one has the following identifications
  \begin{equation}\label{Leg-trf}  
  \mathbf{p}=D_2 S(\widehat{\mathbf{q}},\mathbf{q},\widehat{s}), \qquad s=\widehat{s}+S(\widehat{\mathbf{q}},\mathbf{q},\widehat{s}),\qquad
  \widehat{\mathbf{p}}= - \frac{D_1S(\widehat{\mathbf{q}},\mathbf{q},\widehat{s})}{1+D_3S(\widehat{\mathbf{q}},\mathbf{q},\widehat{s})},
   \end{equation}
where $D_iS$ represents the partial derivative of $S$ with respect to its $i$-th entry.

\subsection{Contact Lagrangian Dynamics}\label{contact:lagrangian}
Now we review the Lagrangian picture of contact systems. In~\cite{de2019singular} we give a more comprehensive description which also covers the case of singular Lagrangians. Let $Q$ be a $n$-dimensional configuration manifold and consider the extended tangent bundle $TQ \times \mathbb{R}$. If $\mathbf{q}=(q^i)$ is a local coordinate system, then the induced coordinates on the $(2n+1)$- dimensional manifold $TQ \times \mathbb{R}$  are  $(\mathbf{q},\dot{\mathbf{q}},s)$. 

\textbf{Herglotz's Action and Herglotz's Equation.}
Consider two points $\mathbf{q}_0$ and $\mathbf{q}_T$ on $Q$, a Lagrangian $L$ on the extended tangent bundle $TQ \times \mathbb{R}$ and the following initial value problem  
\begin{equation}\label{Cauchy}
\dv{s}{t} =L\Big(\mathbf{q}, \dv{\mathbf{q}}{t}, s\Big) , \qquad  s(0)  =s_{0},
\end{equation}
where $\mathbf{q}=\mathbf{q}(t)$ is the curve containing $\mathbf{q}_0$ and $\mathbf{q}_T$ as initial and final points, i.e., $\mathbf{q}(0)=\mathbf{q}_0$ and $\mathbf{q}(T)=\mathbf{q}_T$, with $0\leq t \leq T$. Evidently, the Cauchy problem \eqref{Cauchy} will be different for different curves $\mathbf{q}(t)$. So that a solution $s=s(t)$ of the problem depends on the curve $\mathbf{q}(t)$ substituted in the Lagrangian function. 

Let $\mathcal{C}^\infty(Q)$ be the space of all smooth curves on $Q$ joining $\mathbf{q}_0$ and $\mathbf{q}_T$. This space depends on the end points but we omit this fact in the notation. $\mathcal{C}^\infty(Q)$ is an infinite dimensional manifold. According to the Cauchy problem in \eqref{Cauchy}, for every initial value $s(0)=s_0$, we determine the Herglotz action as a map from the product space $\mathcal{C}^\infty(Q)\times \mathbb{R}$ to the real numbers as follows
\begin{equation}
\mathcal{C}^\infty(Q)\times \mathbb{R}\rightarrow \mathbb{R},\qquad (\mathbf{q},s_0)\mapsto  s(T)-s(0)=\int_{0}^{T} \dv{s}{t} d t=\int_{0}^{T} L \left(\mathbf{q}, \dv{\mathbf{q}}{t}  , s \right) d t.
\end{equation}
The extreme values of the action are the curves $\mathbf{q}=\mathbf{q}(t)$ solving the Herglotz's (generalized Euler-Lagrange) equations \cite{Herglotz,herg2,Guenther}
\begin{equation}\label{eq:herglotz}
    \dv{}{t} L_{\dot{\mathbf{q}}} - L_{\mathbf{q}}= L_sL_{\dot{\mathbf{q}}}.
\end{equation}

\subsection{Contact Hamiltonian Dynamics}

Consider a contact manifold $(M,\eta)$. For a Hamiltonian function $H$, the contact Hamiltonian vector field is defined in terms of the contact one-form $\eta$ as
\begin{equation}
\iota_{X_{H}}\eta =-H,\qquad \iota_{X_{H}}\dd\eta =\dd H-\mathcal{R}(H) \eta,   \label{contact}
\end{equation}%
where $\mathcal{R}$ is the Reeb vector field.  A direct computation determines that a Hamiltonian flow does not preserve the contact one-form
\begin{equation}\label{L-X-eta}
\mathcal{L}_{X_{H}}\eta =
\dd\iota_{X_{H}}\eta+\iota_{X_{H}}\dd\eta= -\mathcal{R}(H)\eta.
\end{equation}
Notice that, according to \eqref{L-X-eta}, the flow of a contact Hamiltonian system preserves the contact structure  which is defined to be the kernel of the contact one-form. $X_{H}$ does not preserve the Hamiltonian function, nor the volume form $\dd\eta^n
\wedge \eta$. Instead, we obtain
\begin{equation}
{\mathcal{L}}_{X_H} \, H = - \mathcal{R}(H) H, \qquad {\mathcal{L}}_{X_H}  \, (\dd \eta^n
\wedge \eta) = - (n+1)  \mathcal{R}(H) \dd \eta^n
\wedge \eta.
\end{equation}
As manifested previously, all contact manifolds locally resemble the extended cotangent bundle $T^*Q\times \mathbb{R}$. In this local view $(\mathbf{q}, \mathbf{p},s)$, the Hamiltonian vector field turns out to be
\begin{equation}\label{hcont2}
X_H = H_\mathbf{p}\cdot \nabla_ \mathbf{q} - \big(H_\mathbf{q}+\mathbf{p}H_s \big )\cdot \nabla_\mathbf{p}+
\big(\mathbf{p}\cdot H_\mathbf{p}-H \big )\nabla_s
\end{equation}
Thus, an integral curve $(\mathbf{q}(t), \mathbf{p}(t), s(t))$ of $X_H$ satisfies the 
contact Hamilton equations 
\begin{equation}\label{hcont3} 
\dv{\mathbf{q}}{t}  =  H_\mathbf{p} , \qquad 
\dv{\mathbf{p}}{t}  =  - \big(H_\mathbf{q}+\mathbf{p}H_s \big ), \qquad 
\dv{s}{t}   = \mathbf{p} \cdot H_\mathbf{p}-H. 
\end{equation}

\textbf{The Legendre Transformation.} The fiber derivative $\mathbb{F} L$ of a Lagrangian function $L$ determines a mapping from the extended tangent bundle to the extended cotangent bundle as 
\begin{equation} 
		\mathbb{F} L:TQ \times \mathbb{R}  \longrightarrow T^*Q \times \mathbb{R}, \qquad (\mathbf{q},\dot{\mathbf{q}},s)\mapsto (\mathbf{q},L_{\dot{\mathbf{q}}},s). 
\end{equation}
See that, for a regular Lagrangian,  $\mathbb{F} L$ is a contactomorphism from the contact manifold $TM\times \mathbb{R} $ equipped with the contact one-form $\eta_L=(\mathbb{F} L)^*\eta$ to the contact manifold $T^*Q \times \mathbb{R}$ equipped with $\eta$ in \eqref{eq:cotangent_contact_structure}. In this situation, we can define the corresponding Hamiltonian function as 
\begin{equation} 
H:T^*Q\times \mathbb{R} \to \mathbb{R},\qquad H=\dot{\mathbf{q}}\cdot L_{\dot{\mathbf{q}}}- L.
\end{equation}
A direct computation proves that the contact Hamilton equations are the generalized Euler-Lagrange equations, they coincide for regular Lagrangian functions. 
For the Legendre transformation of nonregular cases, we cite \cite{esen2021contact,de2019singular,LeonLainz2}.

\subsection{A Continuous Geometric Hamilton-Jacobi Theory on Contact Manifolds}\label{contact:hamiltonian}

We start with the fibration $Q\mapsto \mathbb{R}$ whose first jet bundle is the extended cotangent bundle $T^*Q \times \mathbb{R}$. In this case, the fibration is given by the target map 
\begin{equation}\label{rho}
\rho: T^{*}Q \times \mathbb{R}\longrightarrow Q ,\qquad (\mathbf{z},s)\mapsto \pi_{Q}(\mathbf{z}),
\end{equation}
where $\pi_{Q}$ is the cotangent bundle projection. 
The first prolongation of a smooth function $F$ on $Q$ is a section of the projection $\rho$. We write the first prolongation as 
\begin{equation}\label{gamma}
\gamma =\rm{J}^1  F :Q\longrightarrow T^{*}Q \times \mathbb{R},\qquad \mathbf{q}\mapsto (\mathbf{q}, F_\mathbf{q},F(\mathbf{q})). 
\end{equation}
Notice that $T^*Q \times \mathbb{R}$ is a contact manifold and the image space of the first prolongation $\gamma$ is a Legendrian submanifold of this space. 
 
Consider a section $\gamma$ in form \eqref{gamma} and a contact Hamiltonian vector field $X_H$ given as in \eqref{hcont2}. We define a vector field \begin{equation}\label{X-gamma-con}
X_{H}^{\gamma }:=T\rho\circ X_H \circ \gamma ,
\end{equation}
on the base manifold $Q$ according to the commutativity of the following diagram
\begin{equation}
\xymatrix{   T^{*}Q \times \mathbb{R}
\ar[dd]^{\rho} \ar[rrr]^{X_H}&   & &T( T^{*}Q \times \mathbb{R})\ar[dd]_{T{\rho}}\\
  &  & &\\
Q  \ar@/^2pc/[uu]^{\gamma}\ar[rrr]^{X^{\gamma}_H}&  & & T Q \ar@/_2pc/[uu]_{T\gamma}}
\end{equation}
where $T{\rho}$ is the tangent mapping of the fibration $\rho$ in \eqref{rho}. Next, we state a geometric Hamilton-Jacobi theorem for contact Hamiltonian dynamics.  

\begin{theorem} \label{HJT-Con}
For the first prolongation $\gamma$ of a function $F$ on $Q$, the following two conditions are equivalent:
\begin{enumerate}
\item The vector fields $X_{H}$ and $X^{\gamma}_H$ are $\gamma$-related, that is
\begin{equation}\label{HJT-Con-1-eq-1--}
T\gamma \circ  X_H^{\gamma}  = X_H\circ \gamma 
\end{equation}
where $T\gamma:TQ\mapsto T(T^{*}Q\times \mathbb{R})$ is the tangent mapping of  the section $\gamma$. 
\item The equation 
\begin{equation}\label{HJT-Con-1-eq-2--}
d(H\circ \gamma) = 0
\end{equation}
is fulfilled. 
\end{enumerate}
\end{theorem}
\begin{proof} We prove this theorem in local coordinates. The restriction of the Hamiltonian vector field $X_{H}$ to the image of $\gamma$ is computed to be 
  \begin{equation}\label{Proof-2-}
   \begin{split}
 X_H\circ \gamma (\mathbf{q})&=\Big(\mathbf{q},F_{\mathbf{q}},F(\mathbf{q}) ;H_{\mathbf{p}}\Big\vert_{im (\gamma)},-H_{\mathbf{q}}\Big\vert_{im (\gamma)}
  -H_{s}F_{\mathbf{q}} \Big\vert_{im (\gamma)}, F_{\mathbf{q}} \cdot H_{\mathbf{q}}\Big\vert_{im (\gamma)} -  H\Big\vert_{im (\gamma)}
  \Big). 
   \end{split}
    \end{equation}
Using $T\rho$, we map this vector field down to the tangent bundle of $Q$. This is the projected  vector field \eqref{X-gamma-con} that is
 \begin{equation}
X^{\gamma}_H(\mathbf{q})=   H_{\mathbf{p}}\Big\vert_{im (\gamma)}\cdot \nabla_{\mathbf{q}}.
   \end{equation} 
 On the other hand, the tangent mapping of the section $\gamma$ is computed to be  
 \begin{equation}\label{T-gamma}
 T\gamma(\mathbf{q};\mathbf{\dot{q}})=\big (\mathbf{q},F_{\mathbf{q}}(\mathbf{q}),F(\mathbf{q});\mathbf{\dot{q}},F_{\mathbf{qq}}(\mathbf{q})\mathbf{\dot{q}},F_{\mathbf{q}}(\mathbf{q})\cdot \mathbf{\dot{q}} \big),
 \end{equation}
 where the notation $F_{\mathbf{qq}}(\mathbf{q})\mathbf{\dot{q}}$ stands for the multiplication of the Hessian matrix $F_{\mathbf{qq}}$ with the column vector $\mathbf{\dot{q}}$. 
Accordingly, the left hand side of the equation \eqref{HJT-Con-1-eq-1--} is computed to be
 \begin{equation}\label{Proof-1-}
 \begin{split}
  T\gamma \circ  X^{\gamma}_H(\mathbf{q})&=\big(\mathbf{q},F_{\mathbf{q}}(\mathbf{q}),F(\mathbf{q});H_{\mathbf{p}}\Big\vert_{im (\gamma)},F_{\mathbf{qq}}H_{\mathbf{p}}\Big\vert_{im (\gamma)},F_{\mathbf{q}} \cdot H_{\mathbf{q}}\Big\vert_{im (\gamma)} \big).
 \end{split}
  \end{equation}
To satisfy \eqref{HJT-Con-1-eq-1--}, the expressions \eqref{Proof-2-} and \eqref{Proof-1-} must be same. The first three entries of these local realizations are the same. The first of the fiber variables (that coincides with the fourth entries) are also the same. See that the fifth and the sixth entries of  \eqref{Proof-2-} and \eqref{Proof-1-}  are not equal. The fifth entries are equal if and only if 
 \begin{equation}\label{evo-con-HJ-expli-}
F_{\mathbf{qq}}H_{\mathbf{p}}\Big\vert_{im (\gamma)}+
H_{\mathbf{q}}\Big\vert_{im (\gamma)}
  + H_{s}F_{\mathbf{q}}\Big\vert_{im (\gamma)}=0.
     \end{equation}
     See that, the identity \eqref{evo-con-HJ-expli-} can be compactly written as $d(H\circ \gamma)=0$. This gives that $H$ turns out to be a constant if it is restricted to the image space of the first prolongation $\gamma$. On the other hand, the sixth entries are the same if and only if $H$ is not only constant but further it vanishes on the image space of the first prolongation $\gamma$. This is the second condition \eqref{HJT-Con-1-eq-2--}. The inverse of the assertion is proved by reversing the arguments. This completes the proof.
\end{proof}

There exist alternative versions of the geometric HJ theory for contact manifolds where the base space is considered to be the extended configuration space $Q\times \mathbb{R}$. To check these versions, we refer to \cite{esen2021implicit,LeonLainz4,LeonSar3}.

\section{Discrete Contact Dynamics}\label{DCD-Section}

In this section, we first recall discrete Lagrangian and Hamiltonian dynamics on the contact framework. The definitions and the approach we adopt are the ones displayed in \cite{Bravetti1,Bravetti2,Simoes21,Simoes2020b}. We also wish to cite  \cite{Vermeeren}. Then we introduce the HJ theory for discrete Hamiltonian dynamics. This is the contact version of the approach performed in \cite{OhsawaBlochLeok} for the case of symplectic manifolds. In accordance with this, we state a geometric Hamilton-Jacobi theory. This is a contact generalization of the one presented in \cite{LeonSar3}.

\subsection{Discrete Contact Lagrangian Dynamics}
In the framework of discrete dynamics the role of differentiable curves in a configuration manifold $Q$ is played by the finite sequences of points in $Q$. Therefore, we are interested in the space of sequences consisting of $N+1$ points, denoted by $\mathcal{C}^N(Q)$. It is interesting to note that $\mathcal{C}^N(Q)$ is a manifold isomorphic to $N+1$ number of copies of $Q$, denoted by $Q^{N+1}$. We denote such a sequence by a set $[\mathbf{q}]=\{\mathbf{q}_0,\dots, \mathbf{q}_N\}$ which is a point in $Q^{N+1}$.
Here, $\mathbf{q}_0$ is the initial point of the sequence, whereas $\mathbf{q}_N$ is the end point. We use the subindex $\mathbf{q}_\kappa$ to denote the location of the point in the sequence from $\kappa=0$ to $\kappa=N$. We call $[\mathbf{q}]$ a discrete curve. 

\textbf{Herglotz's Action and Herglotz's Equation.} A discrete Lagrangian in this framework is a real valued function $L_d$ defined on the product space $Q\times Q\times \mathbb{R}$. Comparing with the continuous case, we see that in this geometry the role of the tangent bundle $TQ$ is played by the pair $Q\times Q$. More technically, the discrete Lagrangian $L_d$ is an approximation of the exact contact discrete Lagrangian $L_{d}^{ex}$ which is defined to be
\begin{equation}
L_{d}^{ex}(\mathbf{q}_\kappa,\mathbf{q}_{\kappa+1},s_\kappa)=\int_{t_\kappa}^{t_{\kappa+1}}{L\left(\mathbf{q}, \dv{\mathbf{q}}{t}, s\right)dt}
\end{equation}
where $\mathbf{q}$ is a solution of the continuous generalized Euler-Lagrange equation \eqref{eq:herglotz} with boundary conditions $\mathbf{q}(t_\kappa)=\mathbf{q}_\kappa$ and $\mathbf{q}(t_{\kappa+1})=\mathbf{q}_{\kappa+1}$.
To each discrete curve $[\mathbf{q}]$  and for a fixed initial value point $s_0$, we introduce the discrete Cauchy problem as
\begin{equation}\label{Dis-Cauchy}
 s_{\kappa}-s_{\kappa-1}=L_d(\mathbf{q}_{\kappa-1},\mathbf{q}_{\kappa},s_{\kappa-1}),\qquad s_{0}= s(0),
\end{equation}
where $\kappa$ runs from $1$ to $N$. 
We determine the discrete Herglotz action as a map from the product space $\mathcal{C}^N(Q)\times \mathbb{R}$ to the real numbers that is
\begin{equation}\label{disc-action}
\mathcal{C}^N(Q) \times \mathbb{R}\longrightarrow \mathbb{R},\qquad ([\mathbf{q}],s_0)\mapsto  s_N-s_0=\sum_{\kappa=0}^{N-1}L_d(\mathbf{q}_\kappa,\mathbf{q}_{\kappa+1},s_\kappa).
\end{equation}
The extreme values of the discrete action \eqref{disc-action} are  curves solving the discrete Herglotz's (generalized discrete Euler-Lagrange) equations
\begin{equation}\label{DHE}
\begin{split}
& D_1L_d(\mathbf{q}_\kappa,\mathbf{q}_{\kappa+1},s_\kappa)+\big(1+D_3 L_d(\mathbf{q}_\kappa,\mathbf{q}_{\kappa+1},s_\kappa)\big)D_2 L_d(\mathbf{q}_{\kappa-1},\mathbf{q}_\kappa,s_{\kappa-1})=0
\\
& s_{\kappa}-s_{\kappa-1}=L_d(\mathbf{q}_{\kappa-1},\mathbf{q}_{\kappa},s_{\kappa-1}),\qquad s_{0}= s(0),
 \end{split}
\end{equation}
provided that $1+D_3 L_d\neq 0$. Here, $D_{i}L_d$ refers to the partial derivative of the discrete Lagrangian with respect to its $i$-th entry.

\noindent

\textbf{The Discrete Legendre Transformation.} Assume a discrete Lagrangian function $L_d$ satisfying the condition $1+D_{3}L_{d}\neq 0$. We define the following Legendre transformations from the extended discrete space $Q\times Q  \times \mathbb{R}$ to the extended cotangent bundle $T^*Q \times \mathbb{R}$ as  
\begin{equation}\label{disccontleg}
    \begin{split}
         \mathbb{F}L_{d}^{+}&: Q\times Q  \times \mathbb{R} \longrightarrow T^*Q \times \mathbb{R},
         \\ &\hspace{2cm} (\mathbf{q}_\kappa,\mathbf{q}_{\kappa+1},s_\kappa)\mapsto \big(\mathbf{q}_{\kappa+1},D_{2}L_{d}(\mathbf{q}_\kappa,\mathbf{q}_{\kappa+1},s_\kappa),s_{\kappa}+L_{d}(\mathbf{q}_\kappa,\mathbf{q}_{\kappa+1},s_\kappa)\big)\\
                \mathbb{F}L_{d}^{-}&: Q\times Q  \times \mathbb{R} \longrightarrow T^*Q \times \mathbb{R},
         \\ &\hspace{2cm}
         (\mathbf{q}_\kappa,\mathbf{q}_{\kappa+1},s_\kappa)\mapsto \left(\mathbf{q}_\kappa,-\frac{D_{1}L_{d}(\mathbf{q}_\kappa,\mathbf{q}_{\kappa+1},s_\kappa)}{1+D_{3}L_{d}(\mathbf{q}_\kappa,\mathbf{q}_{\kappa+1},s_\kappa)},s_{\kappa} \right).
    \end{split}
\end{equation}
So that we have two alternative definitions of the momenta given as
\begin{equation}
\mathbf{p}^+_{\kappa-1,\kappa,\kappa}=D_{2}L_{d}(\mathbf{q}_{\kappa-1},\mathbf{q}_{\kappa},s_\kappa), \qquad \mathbf{p}^-_{\kappa,\kappa+1,\kappa}=-\frac{D_{1}L_{d}(\mathbf{q}_\kappa,\mathbf{q}_{\kappa+1},s_\kappa)}{1+D_{3}L_{d}(\mathbf{q}_\kappa,\mathbf{q}_{\kappa+1},s_\kappa)}.
\end{equation} 
In the light of the discrete Herglotz equations \eqref{DHE}, one can directly establish the momentum matching equation
\begin{equation} 
\mathbb{F}L_{d}^{+}(\mathbf{q}_\kappa,\mathbf{q}_{\kappa+1},s_\kappa)=\mathbb{F}L_{d}^{-}(\mathbf{q}_{\kappa+1},\mathbf{q}_{\kappa+2},s_{\kappa+1}).
\end{equation}  
So that, in this case, one can define the momenta 
\begin{equation} 
\mathbf{p}_{\kappa}:=\mathbf{p}^+_{\kappa-1,\kappa,\kappa}=\mathbf{p}^-_{\kappa,\kappa+1,\kappa+1}.
\end{equation}  
A discrete Lagrangian $L_d$ is called regular if the Hessian matrix obtained by taking the partial derivatives of $L_d$ is nondegenerate, that is 
\begin{equation}\label{reg-cond}
\det[D_2D_2L_d(\mathbf{q}_{\kappa-1},\mathbf{q}_{\kappa},s_\kappa)]\neq 0.
\end{equation}
Equivalently, a discrete Lagrangian is a regular Lagrangian if the Legendre transformations in \eqref{disccontleg} become invertible.  

\textbf{Diffeomorphisms.} Consider a regular discrete Lagrangian function $L_d$, referring to the discrete Legendre mappings in \eqref{disccontleg}, one can define a local diffeomorphism on the extended discrete space as 
\begin{equation}\label{Phi}
    \begin{split}
\Phi &: Q\times Q  \times \mathbb{R} \longrightarrow Q\times Q  \times \mathbb{R}, 
         \\ &\hspace{2cm}(\mathbf{q}_\kappa,\mathbf{q}_{\kappa+1},s_\kappa)\mapsto (\mathbf{q}_{\kappa+1},\mathbf{q}_{\kappa+2},s_{\kappa+1}):=(\mathbb{F}L_{d}^{-})^{-1}\circ  \mathbb{F}L_{d}^{+} (\mathbf{q}_\kappa,\mathbf{q}_{\kappa+1},s_\kappa).
             \end{split}
\end{equation}
Similarly, one can define a diffeomorphism on the extended cotangent bundle as follows 
\begin{equation}\label{tilde-Phi}
\tilde{\Phi}:= \mathbb{F}L_{d}^{+} \circ (\mathbb{F}L_{d}^{-})^{-1}.
\end{equation}
A direct comparison shows that the diffeomorphism $\Phi$ given in \eqref{Phi} and  the diffeomorphism $\tilde{\Phi}$ given in \eqref{tilde-Phi} are related by the Legendre transformations 
\begin{equation}
\tilde{\Phi}= \mathbb{F}L_{d}^{+} \circ \Phi \circ (\mathbb{F}L_{d}^{+})^{-1}, \qquad 
\tilde{\Phi}= \mathbb{F}L_{d}^{-} \circ \Phi \circ (\mathbb{F}L_{d}^{-})^{-1}.
\end{equation}
 In order to be more specific about the diffeomorphism we plot the following commutative diagram. 
 \begin{equation}\label{flows}
{\begin{footnotesize}
\xymatrix{ &(\mathbf{q}_{\kappa},\mathbf{q}_{\kappa+1},s_\kappa)
\ar[ddl]_{\mathbb{F}L_{d}^{-}} \ar[ddr]^{\mathbb{F}L_{d}^{+}}\ar[rr]^{\Phi}&  & (\mathbf{q}_{\kappa+1},\mathbf{q}_{\kappa+2},s_{\kappa+1})\ar[ddl]_{\mathbb{F}L_{d}^{-}} \ar[ddr]^{\mathbb{F}L_{d}^{+}} &  \\ \\
 (\mathbf{q}_{\kappa},\mathbf{p}_\kappa,s_\kappa)\ar[rr]_{\tilde{\Phi}}& & (\mathbf{q}_{\kappa+1},\mathbf{p}_{\kappa+1},s_{\kappa+1})\ar[rr]_{\tilde{\Phi}}& & (\mathbf{q}_{\kappa+2},\mathbf{p}_{\kappa+2},s_{\kappa+2})}
 \end{footnotesize}}
\end{equation}  
See that in this diagram in the bottom row we have the flow on the extended discrete cotangent bundle $T^*Q\times \mathbb{R}$,
whereas the top row corresponds with the flow on the extended discrete space $Q\times Q \times \mathbb{R}$. Here, the Legendre transformations establish the equivalency between the flows.

\subsection{Discrete Contact Hamiltonian Dynamics}

To derive Hamiltonian dynamics, we use that a discrete contact Lagrangian is essentially a generating function of type one \cite{Arnold-book} and that
we can apply the defined Legendre transformations to the discrete Lagrangian to find a discrete Hamiltonian \cite{Arnold-book,Goldstein-book}. We start with the definition of momentum $\mathbf{p}_{\kappa+1}$ and the local inversion operation 
\begin{equation}\label{rLT}
\mathbf{p}_{\kappa+1}=D_{2}L_{d}(\mathbf{q}_\kappa,\mathbf{q}_{\kappa+1},s_\kappa), \qquad \mathbf{q}_{\kappa+1}=\Psi(\mathbf{q}_\kappa,\mathbf{p}_{\kappa+1},s_\kappa),
\end{equation}
provided that the regularity condition holds, i.e., the operator $D_{2}D_{2}L_{d}$ does not vanish. 
Referring to these local identifications, we introduce the right discrete Hamiltonian function on the extended cotangent bundle as
\begin{equation}\label{rdh}
 H^{+}_d(\mathbf{q}_\kappa,\mathbf{p}_{\kappa+1},s_\kappa)= \mathbf{p}_{\kappa+1}\cdot \mathbf{q}_{\kappa+1} -L_{d}(\mathbf{q}_\kappa,\mathbf{q}_{\kappa+1},s_\kappa).
\end{equation}
By taking the partial derivatives of the Hamiltonian function $ H^{+}_d$ with respect to its arguments (applying the chain rule referring to \eqref{rLT}), we arrive at the following expressions 
\begin{equation}\label{L-H}
\begin{split}
D_1H^{+}_d(\mathbf{q}_\kappa,\mathbf{p}_{\kappa+1},s_\kappa)&=
-D_1L_{d}(\mathbf{q}_\kappa,\mathbf{q}_{\kappa+1},s_\kappa)\\
D_2H^{+}_d(\mathbf{q}_\kappa,\mathbf{p}_{\kappa+1},s_\kappa)&=\mathbf{q}_{\kappa+1}
\\
D_3H^{+}_d(\mathbf{q}_\kappa,\mathbf{p}_{\kappa+1},s_\kappa)&=
-D_3L_{d}(\mathbf{q}_\kappa,\mathbf{q}_{\kappa+1},s_\kappa).
\end{split}
\end{equation}
Notice that, these equations hold under the identifications \eqref{rLT}.
The second equation in the list determines $\mathbf{q}_{\kappa+1}$. We substitute the first and third identities in \eqref{L-H} into the discrete Euler-Lagrange equation \eqref{DHE}, and using \eqref{rLT}, we get
\begin{equation}
\mathbf{p}_{\kappa}=\frac{D_1H^{+}_d(\mathbf{q}_\kappa,\mathbf{p}_{\kappa+1},s_\kappa)}{1-D_3H^{+}_d(\mathbf{q}_\kappa,\mathbf{p}_{\kappa+1},s_\kappa)}.
\end{equation}
Considering the discrete Cauchy problem \eqref{Dis-Cauchy} as part of the set, we arrive at the following system of equations, which we call right discrete contact Hamilton equations, 
\begin{equation}\label{disconteq}
\begin{split}
\mathbf{q}_{\kappa+1}&=D_2H^{+}_d(\mathbf{q}_\kappa,\mathbf{p}_{\kappa+1},s_\kappa), 
\\
\mathbf{p}_{\kappa}&
= \frac{ D_1H^{+}_d(\mathbf{q}_\kappa,\mathbf{p}_{\kappa+1},s_\kappa)} {(1-D_3H^{+}_d(\mathbf{q}_\kappa,\mathbf{p}_{\kappa+1},s_\kappa))}, \\
 s_{\kappa+1}&=s_\kappa+\mathbf{p}_{\kappa+1}\cdot D_2H^+_d(\mathbf{q}_\kappa,\mathbf{p}_{\kappa+1},s_\kappa)-H^+_d(\mathbf{q}_\kappa,\mathbf{p}_{\kappa+1},s_\kappa).
 \end{split}
\end{equation} 
 
\textbf{Left Hamilton Equations.}
 The left discrete Hamiltonian function on the extended cotangent bundle is
\begin{equation}\label{ldh}
 H^{-}_d(\mathbf{p}_\kappa,\mathbf{q}_{\kappa+1},s_{\kappa+1})= -\mathbf{p}_{\kappa}\cdot \mathbf{q}_{\kappa} -L_{d}(\mathbf{q}_\kappa,\mathbf{q}_{\kappa+1},s_{\kappa+1}).
\end{equation}
By taking the partial derivatives of the Hamiltonian function $ H^{-}_d$, we arrive at the following expressions 
\begin{equation}\label{L-H2}
\begin{split}
D_1H^{-}_d(\mathbf{p}_\kappa,\mathbf{q}_{\kappa+1},s_{\kappa+1})&= -\mathbf{q}_\kappa\\
D_2H^{-}_d(\mathbf{p}_\kappa,\mathbf{q}_{\kappa+1},s_{\kappa+1})&=\mathbf{p}_{\kappa+1}
\\
D_3H^{+}_d(\mathbf{p}_\kappa,\mathbf{q}_{\kappa+1},s_{\kappa+1})&=-D_3L_{d}(\mathbf{q}_\kappa,\mathbf{q}_{\kappa+1},s_{\kappa+1}).
\end{split}
\end{equation}

The second equation in the list determines $\mathbf{p}_{\kappa+1}$. 
\begin{equation}
\mathbf{p}_{\kappa+1}=D_2H^{-}_d(\mathbf{p}_\kappa,\mathbf{q}_{\kappa+1},s_{\kappa+1})=-D_2L_d(\mathbf{q}_\kappa,\mathbf{q}_{\kappa+1},s_{\kappa+1}).
\end{equation}
Considering the discrete Cauchy problem \eqref{Dis-Cauchy} as part of the set, we arrive at the following system of equations, which we call the left discrete contact Hamilton equations, 
\begin{equation}\label{disconteqneg}
\begin{split}
&\mathbf{q}_\kappa=-D_1H^{-}_d(\mathbf{p}_\kappa,\mathbf{q}_{\kappa+1},s_{\kappa+1})\\
&\mathbf{p}_{\kappa+1}=D_2H^{-}_d(\mathbf{p}_\kappa,\mathbf{q}_{\kappa+1},s_{\kappa+1})
\\
& s_{\kappa}=s_{\kappa+1}-\mathbf{p}_{\kappa}\cdot D_1H^{-}_d(\mathbf{p}_\kappa,\mathbf{q}_{\kappa+1},s_{\kappa+1})+H^{-}_d(\mathbf{p}_\kappa,\mathbf{q}_{\kappa+1},s_{\kappa+1}) 
 \end{split}
\end{equation}

From now on, we shall focus on the right discrete dynamics, since everything can be reenacted in terms of left discrete approach straightforwardly.

\subsection{Discrete Contact Hamilton--Jacobi theory}\label{discrete-HJ-Section}

The discrete contact Hamilton--Jacobi equation can be derived in terms of a generating function of a coordinate transformation that trivializes the dynamics, as it is done in the classical continuous theory \cite{Goldstein-book}. For this we need a generating function to introduce a contact diffeomorphism from the contact manifold $T^*Q\times \mathbb{R}$ to the contact manifold $T^*\widehat{Q} \times \widehat{\mathbb{R}}$ that is 
\begin{equation}\label{trsformation}
 T^*Q\times \mathbb{R}\longrightarrow   T^*\widehat{Q} \times \widehat{\mathbb{R}},\qquad  (\mathbf{q}_\kappa,\mathbf{p}_\kappa,s_\kappa)\rightarrow (\widehat{\mathbf{q}}_\kappa,\widehat{\mathbf{p}}_\kappa,\widehat{s}_\kappa).
\end{equation} 
Recall from \eqref{Leg-trf} that, in Section \ref{Sec-Cont-Man}, we have obtained a generating function defined on the product space $\widehat{Q} \times Q \times \widehat{\mathbb{R}}$ realizing the transformation. The main result for the generating function is described in the next theorem, following the lines of \cite{OhsawaBlochLeok} in the symplectic case. 
\begin{theorem}\label{theorem-dHJ}
 Consider the right discrete contact Hamilton equations \eqref{disconteq} and a discrete phase space $(\widehat{\mathbf{q}}_\kappa,\mathbf{q}_\kappa,s_\kappa)$. Consider also a transformation \eqref{trsformation} 
 that satisfies:
 \begin{enumerate}
  \item The old and new coordinates are related by a generating function $S^\kappa:\widehat{Q}\times Q\times \widehat{\mathbb{R}}\rightarrow \mathbb{R}$
  of the type
  \begin{equation}\label{genfunction}
 \mathbf{p}_\kappa=D_2 S^\kappa(\widehat{\mathbf{q}}_\kappa,\mathbf{q}_\kappa,\widehat{s}_\kappa), \quad s_\kappa=\widehat{s}_\kappa+S^\kappa(\widehat{\mathbf{q}}_\kappa,\mathbf{q}_\kappa,\widehat{s}_\kappa),\quad
  \widehat{\mathbf{p}}_\kappa= - \frac{D_1S^\kappa(\widehat{\mathbf{q}}_\kappa,\mathbf{q}_\kappa,\widehat{s}_\kappa)}{1+D_3S^\kappa(\widehat{\mathbf{q}}_\kappa,\mathbf{q}_\kappa,\widehat{s}_\kappa)}.
  \end{equation}
\item The dynamics in the new coordinates $(\widehat{\mathbf{q}}_\kappa,\widehat{\mathbf{p}}_\kappa,\widehat{s}_\kappa)$ is rendered trivial, i.e., 
\begin{equation}\label{tri-disc}
(\widehat{\mathbf{q}}_\kappa,\widehat{\mathbf{p}}_\kappa,\widehat{s}_\kappa)=(\widehat{\mathbf{q}}_{\kappa+1},\widehat{\mathbf{p}}_{\kappa+1},\widehat{s}_{\kappa+1}).
\end{equation}
  \end{enumerate}
  Then, the set of functions $\{S^\kappa\}$ satisfies the right discrete contact Hamilton--Jacobi equation:
   \begin{equation}\label{firstdiscrete}
   \begin{split}
   S^{\kappa+1}(\widehat{\mathbf{q}}_{0},\mathbf{q}_{\kappa +1},\widehat{s}_{0})&-S^{\kappa}(\widehat{\mathbf{q}}_{0},\mathbf{q}_{\kappa},\widehat{s}_{0})-D_2 S^{\kappa+1}(\widehat{\mathbf{q}}_{0},\mathbf{q}_{\kappa+1},\widehat{s}_{0})\cdot \mathbf{q}_{\kappa+1} \\ &
   +H^+_d(\mathbf{q}_\kappa,D_2 S^{\kappa+1}(\widehat{\mathbf{q}}_0,\mathbf{q}_{\kappa+1},\widehat{s}_0),\widehat{s}_0+S^\kappa(\widehat{\mathbf{q}}_0,\mathbf{q}_\kappa,\widehat{s}_0))
   = 0.
   \end{split}
\end{equation}
\end{theorem}
\begin{proof} We establish the proof in four steps. 

\textit{Step 1.} Consider the right discrete contact Hamilton equations \eqref{disconteq} in the new variables $(\widehat{\mathbf{q}}_\kappa,\widehat{\mathbf{p}}_\kappa,\widehat{s}_\kappa)$, that is
\begin{equation}\label{disconteq2}
   \begin{split}
\widehat{\mathbf{q}}_{\kappa+1}&=D_2\widehat{H}^{+}_d(\widehat{\mathbf{q}}_\kappa,\widehat{\mathbf{p}}_{\kappa+1},\widehat{s}_\kappa), \\ \widehat{\mathbf{p}}_{\kappa}&
= D_1\widehat{H}^{+}_d(\widehat{\mathbf{q}}_\kappa,\widehat{\mathbf{p}}_{\kappa+1},\widehat{s}_\kappa)/\big(1-D_3\widehat{H}^{+}_d(\widehat{\mathbf{q}}_\kappa,\widehat{\mathbf{p}}_{\kappa+1},\widehat{s}_\kappa)\big), \\
 \widehat{s}_{\kappa+1}&=\widehat{s}_\kappa+\widehat{\mathbf{p}}_{\kappa+1}\cdot D_2\widehat{H}^+_d(\widehat{\mathbf{q}}_\kappa,\widehat{\mathbf{p}}_{\kappa+1},\widehat{s}_\kappa)-\widehat{H}^+_d(\widehat{\mathbf{q}}_\kappa,\widehat{\mathbf{p}}_{\kappa+1},\widehat{s}_\kappa).
 \end{split}
\end{equation}
These equations can be recast in the form of a total derivative of the right discrete contact Hamiltonian in the following way
\begin{equation}\label{totalderivativeham}
   \begin{split}
    \dd{\widehat{H}^{+}_{d}}(\widehat{\mathbf{q}}_\kappa,\widehat{\mathbf{p}}_{\kappa+1},\widehat{s}_\kappa) &=D_1{\widehat{H}^{+}_{d}}(\widehat{\mathbf{q}}_\kappa,\widehat{\mathbf{p}}_{\kappa+1},\widehat{s}_\kappa)\cdot  \dd \widehat{\mathbf{q}}_\kappa+D_2{\widehat{H}^{+}_{d}}(\widehat{\mathbf{q}}_\kappa,\widehat{\mathbf{p}}_{\kappa+1},\widehat{s}_\kappa)\cdot  \dd \widehat{\mathbf{p}}_{\kappa+1}
    \\
    &\hspace{6cm} +D_3{\widehat{H}^{+}_{d}}(\widehat{\mathbf{q}}_\kappa,\widehat{\mathbf{p}}_{\kappa+1},\widehat{s}_\kappa) \dd \widehat{s_\kappa}
         \\&=\widehat{\mathbf{p}}_\kappa\cdot \dd \widehat{\mathbf{q}}_\kappa+\widehat{\mathbf{q}}_{\kappa+1}\cdot \dd \widehat{\mathbf{p}}_{\kappa+1}+D_3{\widehat{H}^{+}_{d}}(\widehat{\mathbf{q}}_\kappa,\widehat{\mathbf{p}}_{\kappa+1},\widehat{s}_\kappa) \widehat{\eta}_\kappa
     \end{split}
\end{equation}
where we have employed \eqref{disconteq2} in the total derivative. More explicitly we have replaced $D_1{\widehat{H}^{+}_{d}}$ from the second equation in \eqref{disconteq2}.  Notice that we have used the following notation for the discrete contact form
\begin{equation}
\widehat{\eta}_\kappa=
\dd \widehat{s}_\kappa-\widehat{\mathbf{p}}_{\kappa}\cdot \dd \widehat{\mathbf{q}}_\kappa. 
\end{equation}  

\textit{Step 2.} Start with the generating function $S=S^\kappa(\widehat{\mathbf{q}}_\kappa,\mathbf{q}_\kappa,\widehat{s}_\kappa)$ satisfying  \eqref{genfunction}. The total derivative of the generating function as
\begin{equation}\label{totalderivativeS}
  \begin{split}
 \dd S^\kappa(\widehat{\mathbf{q}}_\kappa,\mathbf{q}_\kappa,\widehat{s}_\kappa)&=D_1S^\kappa (\widehat{\mathbf{q}}_\kappa,\mathbf{q}_\kappa,\widehat{s}_\kappa) \cdot \dd \widehat{\mathbf{q}}_\kappa+D_2S^\kappa (\widehat{\mathbf{q}}_\kappa,\mathbf{q}_\kappa,\widehat{s}_\kappa) \cdot d\mathbf{q}_\kappa+D_3S^\kappa (\widehat{\mathbf{q}}_\kappa,\mathbf{q}_\kappa,\widehat{s}_\kappa) \dd \widehat{s}_\kappa
 \\
 &=-\widehat{\mathbf{p}}_{\kappa}\cdot \dd \widehat{\mathbf{q}}_\kappa+\mathbf{p}_{\kappa}\cdot d\mathbf{q}_\kappa+D_3S^\kappa(\widehat{\mathbf{q}}_\kappa,\mathbf{q}_\kappa,\widehat{s}_\kappa) \widehat{\eta}_\kappa
   \end{split}
 \end{equation}
 in which we can introduce $D_1S^k$ and $D_2S^k$ from \eqref{genfunction}.
Recursively, we can write the expression for $\dd S^{\kappa+1}$, as:
 \begin{equation}\label{Sj1}
     \dd S^{\kappa +1} (\widehat{\mathbf{q}}_{\kappa +1},\mathbf{q}_{\kappa +1},\widehat{s}_{\kappa +1})=-\widehat{\mathbf{p}}_{\kappa +1}\cdot \dd \widehat{\mathbf{q}}_{\kappa +1}+\mathbf{p}_{\kappa +1}\cdot d\mathbf{q}_{\kappa +1} + D_3S^{\kappa +1}(\widehat{\mathbf{q}}_{\kappa +1},\mathbf{q}_{\kappa +1},\widehat{s}_{\kappa +1}) \widehat{\eta}_{\kappa +1}
 \end{equation}
 where $\widehat{\eta}_{\kappa +1}=
\dd \widehat{s}_{\kappa +1}-\widehat{\mathbf{p}}_{\kappa +1}\cdot \dd \widehat{\mathbf{q}}_{\kappa +1}$.

 Coming back to the expression of the total derivative of the Hamiltonian \eqref{totalderivativeham}, realize that the second term $\widehat{\mathbf{q}}_{\kappa+1}\cdot \dd \widehat{\mathbf{p}}_{\kappa+1}$ can be written as $d(\widehat{\mathbf{q}}_{\kappa+1}\cdot \widehat{\mathbf{p}}_{\kappa+1}) -\widehat{\mathbf{p}}_{\kappa+1}\cdot \dd \widehat{\mathbf{q}}_{\kappa+1}$. Then we substitute the expressions \eqref{totalderivativeS} and \eqref{Sj1} into the total derivative of the Hamiltonian function. In the light of these, we can continue the calculation \eqref{totalderivativeham} as follows
 \begin{equation}\label{totalderivativeham2}
   \begin{split}
 \dd {\widehat{H}^{+}_{d}}(\widehat{\mathbf{q}}_\kappa,&\widehat{\mathbf{p}}_{\kappa+1},\widehat{s}_\kappa)=\widehat{\mathbf{p}}_\kappa\cdot \dd \widehat{\mathbf{q}}_\kappa+\widehat{\mathbf{q}}_{\kappa+1}\cdot \dd \widehat{\mathbf{p}}_{\kappa+1}+D_3{\widehat{H}^{+}_{d}}(\widehat{\mathbf{q}}_\kappa,\widehat{\mathbf{p}}_{\kappa+1},\widehat{s}_\kappa) \widehat{\eta}_\kappa
 \\&
 =\widehat{\mathbf{p}}_\kappa\cdot \dd \widehat{\mathbf{q}}_\kappa -\widehat{\mathbf{p}}_{\kappa+1}\cdot \dd \widehat{\mathbf{q}}_{\kappa+1}+\dd(\widehat{\mathbf{q}}_{\kappa+1}\cdot \widehat{\mathbf{p}}_{\kappa+1})+D_3{\widehat{H}^{+}_{d}}(\widehat{\mathbf{q}}_\kappa,\widehat{\mathbf{p}}_{\kappa+1},\widehat{s}_\kappa) \widehat{\eta}_\kappa
 \\& = \mathbf{p}_{\kappa}\cdot \dd \mathbf{q}_{\kappa } + D_3S^{\kappa }(\widehat{\mathbf{q}}_{\kappa },\mathbf{q}_{\kappa },\widehat{s}_{\kappa }) \widehat{\eta}_{\kappa } - \dd S^{\kappa}  \\&\qquad -\big( \mathbf{p}_{\kappa +1}\cdot \dd \mathbf{q}_{\kappa +1} +  D_3S^{\kappa +1}(\widehat{\mathbf{q}}_{\kappa +1},\mathbf{q}_{\kappa +1},\widehat{s}_{\kappa +1}) \widehat{\eta}_{\kappa +1} - \dd S^{\kappa+1} \big) 
 \\&\qquad +
 \dd(\widehat{\mathbf{q}}_{\kappa+1}\cdot \widehat{\mathbf{p}}_{\kappa+1})+D_3{\widehat{H}^{+}_{d}}(\widehat{\mathbf{q}}_\kappa,\widehat{\mathbf{p}}_{\kappa+1},\widehat{s}_\kappa) \widehat{\eta}_\kappa.
  \end{split}
 \end{equation}

\textit{Step 3.}   We take the exterior derivative of the second identity in \eqref{genfunction}. 
 See that 
 \begin{equation}
   \begin{split}
ds_\kappa&=\dd \widehat{s}_\kappa+\dd S^\kappa(\widehat{\mathbf{q}}_\kappa,\mathbf{q}_\kappa,\widehat{s}_\kappa)\\&=\dd \widehat{s}_\kappa + D_1S^\kappa(\widehat{\mathbf{q}}_\kappa,\mathbf{q}_\kappa,\widehat{s}_\kappa)\cdot \dd \widehat{\mathbf{q}}_\kappa+ 
D_2S^\kappa(\widehat{\mathbf{q}}_\kappa,\mathbf{q}_\kappa,\widehat{s}_\kappa)\cdot \dd \mathbf{q}_\kappa + D_3S^\kappa(\widehat{\mathbf{q}}_\kappa,\mathbf{q}_\kappa,\widehat{s}_\kappa)\cdot \dd \widehat{s}_\kappa
\\&=\dd \widehat{s}_\kappa - (1+D_3S^\kappa(\widehat{\mathbf{q}}_\kappa,\mathbf{q}_\kappa,\widehat{s}_\kappa))\widehat{\mathbf{p}}_\kappa \cdot \dd \widehat{\mathbf{q}}_\kappa + \mathbf{p}_\kappa  \cdot \dd \mathbf{q}_\kappa + D_3S^\kappa(\widehat{\mathbf{q}}_\kappa,\mathbf{q}_\kappa,\widehat{s}_\kappa)\cdot \dd \widehat{s}_\kappa
\\&= \big(1 + D_3S^\kappa(\widehat{\mathbf{q}}_\kappa,\mathbf{q}_\kappa,\widehat{s}_\kappa)\big) \widehat{\eta}_\kappa+ \mathbf{p}_\kappa  \cdot \dd \mathbf{q}_\kappa.
  \end{split}
  \end{equation}
  So that we arrive at two expressions. One is the relationship between the contact forms in terms of the generating function
   \begin{equation}
   \eta_\kappa=\big(1 + D_3S^\kappa(\widehat{\mathbf{q}}_\kappa,\mathbf{q}_\kappa,\widehat{s}_\kappa)\big) \widehat{\eta}_\kappa.
    \end{equation}
Additionally, we have the followings identities
\begin{equation}
\begin{split}
&\mathbf{p}_\kappa  \cdot \dd \mathbf{q}_\kappa +
D_3S^\kappa(\widehat{\mathbf{q}}_\kappa,\mathbf{q}_\kappa,\widehat{s}_\kappa) \widehat{\eta}_\kappa  = ds_\kappa - \widehat{\eta}_\kappa,
\\
&\mathbf{p}_{\kappa+1}  \cdot \dd \mathbf{q}_{\kappa+1} +
D_3S^\kappa(\widehat{\mathbf{q}}_{\kappa+1},\mathbf{q}_{\kappa+1},\widehat{s}_{\kappa+1}) \widehat{\eta}_{\kappa+1}  = \dd s_{\kappa+1} - \widehat{\eta}_{\kappa+1}.
\end{split}
\end{equation}
We continue the calculation \eqref{totalderivativeham2} by  substituting these equations. Accordingly, we have that 
 \begin{equation}\label{pre-HJ}
   \begin{split}
 \dd {\widehat{H}^{+}_{d}}(\widehat{\mathbf{q}}_\kappa,\widehat{\mathbf{p}}_{\kappa+1},\widehat{s}_\kappa)&=
 \mathbf{p}_{\kappa}\cdot \dd \mathbf{q}_{\kappa } + D_3S^{\kappa }(\widehat{\mathbf{q}}_{\kappa },\mathbf{q}_{\kappa },\widehat{s}_{\kappa }) \widehat{\eta}_{\kappa } - \dd S^{\kappa}  \\&\qquad -\big( \mathbf{p}_{\kappa +1}\cdot \dd \mathbf{q}_{\kappa +1} +  D_3S^{\kappa +1}(\widehat{\mathbf{q}}_{\kappa +1},\mathbf{q}_{\kappa +1},\widehat{s}_{\kappa +1}) \widehat{\eta}_{\kappa +1} - \dd S^{\kappa+1} \big) 
 \\&\qquad +
 d(\widehat{\mathbf{q}}_{\kappa+1}\cdot \widehat{\mathbf{p}}_{\kappa+1})+D_3{\widehat{H}^{+}_{d}}(\widehat{\mathbf{q}}_\kappa,\widehat{\mathbf{p}}_{\kappa+1},\widehat{s}_\kappa) \widehat{\eta}_\kappa.
 \\& =\dd s_\kappa - \widehat{\eta}_\kappa - \dd S^{\kappa} - (\dd s_{\kappa+1} - \widehat{\eta}_{\kappa+1} - \dd S^{\kappa+1} ) +
\dd (\widehat{\mathbf{q}}_{\kappa+1}\cdot \widehat{\mathbf{p}}_{\kappa+1})\\&\qquad +D_3{\widehat{H}^{+}_{d}}(\widehat{\mathbf{q}}_\kappa,\widehat{\mathbf{p}}_{\kappa+1},\widehat{s}_\kappa) \widehat{\eta}_\kappa.
  \end{split}
\end{equation}
 Since on image space, i.e, the space with coordinates $(\widehat{\mathbf{q}},\widehat{\mathbf{p}},\widehat{s})$,  the dynamics is rendered trivial \eqref{tri-disc}, we take the Hamiltonian function as 
 \begin{equation}
 \widehat{H}^+_d(\widehat{\mathbf{q}}_\kappa,\widehat{\mathbf{p}}_{\kappa+1},\widehat{s}_\kappa)=\widehat{\mathbf{q}}_\kappa\cdot \widehat{\mathbf{p}}_{\kappa+1}.
 \end{equation}
This gives that the Hamiltonian function is independent of $\widehat{s}_\kappa$ so that $D_3{\widehat{H}^{+}_{d}}$ vanishes identically. We substitute this into \eqref{pre-HJ} and arrive at 
 \begin{equation}\label{preHJ-1}
 \begin{split}
\dd {\widehat{H}^{+}_{d}}& =\dd s_\kappa - \widehat{\eta}_\kappa - \dd S^{\kappa} - (\dd s_{\kappa+1} - \widehat{\eta}_{\kappa+1} - \dd S^{\kappa+1} ) +
 \dd (\widehat{\mathbf{q}}_{\kappa+1}\cdot \widehat{\mathbf{p}}_{\kappa+1})  +D_3{\widehat{H}^{+}_{d}} \widehat{\eta}_\kappa.
 \\
 \dd (\widehat{\mathbf{q}}_\kappa\cdot \widehat{\mathbf{p}}_{\kappa+1})&= \dd s_\kappa - \widehat{\eta}_\kappa - \dd S^{\kappa} - (ds_{\kappa+1} - \widehat{\eta}_{\kappa+1} - \dd S^{\kappa+1} ) +
 \dd (\widehat{\mathbf{q}}_{\kappa+1}\cdot \widehat{\mathbf{p}}_{\kappa+1}).
 \end{split}
\end{equation}

\textit{Step 4.} Since the dynamics is trivial on the image space, the discrete contact forms coincide \begin{equation}\label{eta-cons}
\widehat{\eta}_\kappa=\widehat{\eta}_{\kappa+1}
\end{equation}
and so do the coupling functions
 \begin{equation}\label{cross-cons}
 \widehat{\mathbf{q}}_\kappa\cdot \widehat{\mathbf{p}}_{\kappa+1}=\widehat{\mathbf{q}}_{\kappa+1}\cdot \widehat{\mathbf{p}}_{\kappa+1}=\widehat{\mathbf{q}}_0\cdot \widehat{\mathbf{p}}_{0}.
 \end{equation}
Now, we take the exterior derivative of the third equation in \eqref{disconteq} and replace the second equation in \eqref{disconteq}. This gives
\begin{equation}\label{ds-ds}
\begin{split}
\dd s_{\kappa+1}-\dd s_\kappa&=\dd (\mathbf{p}_{\kappa+1}\cdot \mathbf{q}_{\kappa+1})-\dd H^+_d(\mathbf{q}_\kappa,\mathbf{p}_{\kappa+1},s_\kappa)\\&=
\dd (D_2 S^{\kappa+1}(\widehat{\mathbf{q}}_{\kappa+1},\mathbf{q}_{\kappa+1},\widehat{s}_{\kappa+1})\cdot \mathbf{q}_{\kappa+1})-\dd H^+_d(\mathbf{q}_\kappa,\mathbf{p}_{\kappa+1},s_\kappa)
\end{split}
\end{equation}
We collect all the equations in \eqref{eta-cons}, \eqref{cross-cons} and \eqref{ds-ds} and substitute them into \eqref{preHJ-1}. These read the following
\begin{equation}
    \begin{split}
       \dd (\widehat{\mathbf{q}}_\kappa\cdot \widehat{\mathbf{p}}_{\kappa+1})&= \dd s_\kappa - \widehat{\eta}_\kappa - \dd S^{\kappa} - (\dd s_{\kappa+1} - \widehat{\eta}_{\kappa+1} - \dd S^{\kappa+1} ) +
\dd (\widehat{\mathbf{q}}_{\kappa+1}\cdot \widehat{\mathbf{p}}_{\kappa+1})  \\
 \dd (\widehat{\mathbf{q}}_0 \cdot \widehat{\mathbf{p}}_{0})&= \dd  s_\kappa - \widehat{\eta}_\kappa - \dd S^{\kappa} - (\dd s_{\kappa+1} - \widehat{\eta}_{\kappa+1} - \dd S^{\kappa+1} ) +
\dd (\widehat{\mathbf{q}}_{0}\cdot \widehat{\mathbf{p}}_{0})
\\
 0&=\dd s_\kappa-\dd s_{\kappa+1}+\dd S^{\kappa+1}-\dd S^{\kappa}
 \\
  0&=-\dd (D_2 S^{\kappa+1}(\widehat{\mathbf{q}}_{\kappa+1},\mathbf{q}_{\kappa+1},\widehat{s}_{\kappa+1})\cdot \mathbf{q}_{\kappa+1})+\dd H^+_d
  (\mathbf{q}_\kappa,\mathbf{p}_{\kappa+1},s_\kappa)+\dd S^{\kappa+1}-\dd S^{\kappa}
   \\
  0&=\dd \big(S^{\kappa+1}-S^{\kappa}-D_2 S^{\kappa+1}(\widehat{\mathbf{q}}_{\kappa+1},\mathbf{q}_{\kappa+1},\widehat{s}_{\kappa+1})\cdot \mathbf{q}_{\kappa+1} +H^+_d(\mathbf{q}_\kappa,\mathbf{p}_{\kappa+1},s_\kappa)\big).
     \end{split}
\end{equation}
So that we have obtained a Hamilton-Jacobi equation \eqref{firstdiscrete} for discrete contact Hamiltonian dynamics.
\end{proof}

Now, fix the initial point $(\widehat{\mathbf{q}}_{0},\mathbf{q}_{0},\widehat{s}_{0})$. For the generating function \eqref{genfunction} presented in the previous subsection, we introduce a new notation as follows
\begin{equation}\label{newnotation}
S^{\kappa}_d(\mathbf{q}_{\kappa}):=\widehat{s}_{0}+S^{\kappa}(\widehat{\mathbf{q}}_{0},\mathbf{q}_{\kappa},\widehat{s}_{0}).
\end{equation}
This denotation determines a function from $Q$ to the real numbers. In this notation, the discrete contact Hamilton-Jacobi equation \eqref{firstdiscrete} turns out to be
  \begin{equation}\label{firstdiscrete-1}
   \begin{split}
   S^{\kappa+1}_d(\mathbf{q}_{\kappa +1})-S^{\kappa}_d(\mathbf{q}_{\kappa})-\dd S^{\kappa+1}_d(\mathbf{q}_{\kappa+1})\cdot \mathbf{q}_{\kappa+1} +H^+_d(\mathbf{q}_\kappa,\dd S^{\kappa+1}_d(\mathbf{q}_{\kappa+1}),S^{\kappa}_d(\mathbf{q}_{\kappa}))
   = 0.
   \end{split}
\end{equation}

 One can obtain a similar result in terms of the left discrete Hamiltonian. Using the new notation \eqref{newnotation}, we present the left discrete contact Hamilton--Jacobi equation
  \begin{equation}\label{firstdiscreteleft}
   \begin{split}
   S^{\kappa+1}_d(\mathbf{q}_{\kappa +1})-S^{\kappa}_d(\mathbf{q}_{\kappa})+\dd S^{\kappa}_d(\mathbf{q}_{\kappa})\cdot \mathbf{q}_{\kappa} +H^-_d(\mathbf{q}_{\kappa+1},\dd S^{\kappa}_d(\mathbf{q}_{\kappa}),S^{\kappa+1}_d(\mathbf{q}_{\kappa+1}))
   = 0.
   \end{split}
\end{equation}

\subsection{A Geometric Discrete Hamilton-Jacobi Theory on Contact Manifolds}\label{gHJ-Section}

Having obtained a discrete contact Hamilton--Jacobi equation \eqref{firstdiscrete-1} in the previous section, it is easy to see that that the differential of the discrete generating function \eqref{newnotation} is precisely the discrete version of the continuous section $\gamma$ in Theorem \label{HJT-Con}.
Here, we aim at interpreting the discrete contact Hamilton--Jacobi equation \eqref{firstdiscrete-1} in terms of discrete flows.

\textbf{Projection of Discrete Flow.} Consider a Hamiltonian function  $H^{+}_d=H^{+}_d(\mathbf{q}_\kappa,\mathbf{p}_{\kappa+1},s_\kappa)$. Referring to the commutativity of diagram \eqref{flows} and in the light of the discrete contact Hamiltonian dynamics  \eqref{disconteq}, we compute the Hamiltonian flow on the extended cotangent bundle 
as
\begin{equation}
\tilde{\Phi}:
\Big(\mathbf{q}_\kappa,\frac{D_1H^{+}_d }
  {1-D_3H^{+}_d},s_\kappa\Big)\longrightarrow \Big(D_2H^{+}_d,\mathbf{p}_{\kappa+1},s_\kappa+\mathbf{p}_{\kappa+1}\cdot D_2H^+_d -H^+_d \Big)
\end{equation}
Recall the projection $\rho:T^*Q\times \mathbb{R} \mapsto Q\times \mathbb{R}$  given in \eqref{rho}. For a function $S_d=S_d(\mathbf{q}_\kappa)$ on the base space, we define local sections of the projection $\rho$ as 
\begin{equation}\label{gamma-a}
\begin{split}
\gamma^\kappa &:Q \longrightarrow T^*Q\times \mathbb{R},\qquad \mathbf{q}_\kappa \mapsto (\mathbf{q}_\kappa,\dd S_d^\kappa(\mathbf{q}_\kappa),S_d^\kappa(\mathbf{q}_\kappa)),
\\
\gamma^{\kappa+1} &:Q \longrightarrow T^*Q\times \mathbb{R},\qquad  \mathbf{q}_{\kappa+1} \mapsto (\mathbf{q}_{\kappa+1},\dd S_d^\kappa(\mathbf{q}_{\kappa+1}),S_d^{\kappa+1}(\mathbf{q}_{\kappa+1})).
\end{split}
\end{equation}
According to the commutative diagram
\begin{equation}
\xymatrix{ (\mathbf{q}_\kappa,\mathbf{p}_\kappa,s_\kappa)
\ar[dd] ^{\rho^*} \ar[rrr]^{\tilde{\Phi}}&   & & (\mathbf{q}_{\kappa+1},\mathbf{p}_{\kappa+1},s_{\kappa+1})\ar[dd]^{\rho^*}\\
  &  & &\\
 (\mathbf{q}_\kappa,s_\kappa)\ar@/^2pc/[uu]^{\gamma^\kappa}\ar[rrr]^{\tilde{\Phi}^{\gamma^\kappa}}&  & & (\mathbf{q}_{\kappa+1},s_{\kappa+1})\ar@/_2pc/[uu]_{\gamma^{\kappa+1}}}
\end{equation}
we pull down the flow $\tilde{\Phi}$ to the base manifold $Q$ and obtain the following projected discrete flow
\begin{equation}
\tilde{\Phi}^{\gamma}:Q  \longrightarrow Q  ,\qquad  \mathbf{q}_{\kappa} \mapsto \mathbf{q}_{\kappa +1}=D_2H^{+}_d \big(\mathbf{q}_\kappa, \dd s_d^{\kappa+1} (\mathbf{q}_{\kappa+1}),S_d^{\kappa}(\mathbf{q}_{\kappa})\big).
\end{equation}
This procedure reduces the number of dependent variables to the coordinates $\mathbf{q}$ on the base manifold $Q$. By means of the section $\gamma$, one can lift a solution of the projected flow to the extended cotangent bundle level. The lifted solution becomes a solution to the Hamiltonian flow if and only if the  commutation relation 
\begin{equation}\label{commit}
\tilde{\Phi}\circ \gamma^{\kappa} = \gamma^{\kappa+1}\circ \tilde{\Phi}^{\gamma}.
\end{equation}
holds. 
\begin{theorem}\label{gHJ-discrete}
For a section $\gamma$ admitting the local forms in  \eqref{gamma-a}, the following conditions are equivalent: 
\begin{enumerate}
   \item The flows $\tilde{\Phi}$ and $\tilde{\Phi}^{\gamma}$ commute, i.e., $\tilde{\Phi}\circ \gamma^{\kappa} = \gamma^{\kappa+1}\circ \tilde{\Phi}^{\gamma}.$
    \item $S$ solves the HJ equation \eqref{firstdiscrete-1}.
\end{enumerate}  
\end{theorem}
\begin{proof}
 In terms of the local coordinates, the commutation condition (that is the first of the condition in the statement of the theorem) gives rise to the following equations
\begin{equation}\label{condi}
\begin{split}
S_d^{\kappa+1}(\mathbf{q}_{\kappa+1})&=S_d^{\kappa}(\mathbf{q}_{\kappa})+\dd s_d^{\kappa+1}(\mathbf{q}_{\kappa+1})\cdot \mathbf{q}_{\kappa+1} - H^{+}_d \big(\mathbf{q}_\kappa, \dd s_d^{\kappa+1}(\mathbf{q}_{\kappa+1}),S_d^{\kappa}(\mathbf{q}_{\kappa})\big),
\\
\dd S_d^{\kappa}(\mathbf{q}_\kappa)&= {D_1H^{+}_d \big(\mathbf{q}_\kappa, \dd S_d^{\kappa+1}(\mathbf{q}_{\kappa+1}) ,S_d^{\kappa}(\mathbf{q}_{\kappa})\big)}\big / \big(
   1-D_3H^{+}_d\big(\mathbf{q}_\kappa, \dd S_d^{\kappa+1}(\mathbf{q}_{\kappa+1}) ,S_d^{\kappa}(\mathbf{q}_{\kappa})\big)\big).
\end{split}
\end{equation}
A direct comparison gives that the second condition in \eqref{condi} is the infinitesimal version of the first one. Indeed, if we take the derivative of the first equation in \eqref{condi}  with respect to $\mathbf{q}_\kappa$, we have that 
\begin{equation}
0=\dd S_d^{\kappa}(\mathbf{q}_{\kappa}) - D_1H^{+}_d \big(\mathbf{q}_\kappa, \dd S_d^{\kappa+1}(\mathbf{q}_{\kappa+1}),S_d^{\kappa}(\mathbf{q}_{\kappa})\big)  - D_3H^{+}_d \big(\mathbf{q}_\kappa, \dd S_d^{\kappa+1}(\mathbf{q}_{\kappa+1}),S_d^{\kappa}(\mathbf{q}_{\kappa})\big)\dd S_d^{\kappa}(\mathbf{q}_{\kappa})
\end{equation}
By collecting all the terms involving $\dd S_d^{\kappa}$ in the left hand side, it is immediate to see that it is precisely the second condition in \eqref{condi}. We remark that the first equation is the discrete Hamilton-Jacobi equation \eqref{firstdiscrete-1}. The inverse of the assertion is straightforwardly proved by reversing the arguments.
\end{proof}

\section{An Application: The Parachute Equation}\label{Application}

In \cite{gaset2020new} the so-called parachute equation was introduced. This equation  describes the vertical motion of a particle falling in a fluid under the action of constant gravity. The friction is modeled by the drag equation and it is proportional to the square of the velocity.
Here we present the Hamiltonian associated with the discrete parachute equation on
the extended phase space 
$T^*Q\times \mathbb{R}$
where $Q$ is a one-dimensional manifold with dynamic variable $q$ that is discretized as: $q\rightarrow (q_1,\dots,q_N)$, so is the associated momentum $p$ as $p\rightarrow (p_1,\dots,p_N)$. So, the discrete right contact Hamiltonian for the parachute reads:
\begin{equation}\label{eq:parachute_hamiltoniam2}
    H^+_d(q_\kappa,p_{\kappa+1},s_\kappa) = \frac{1}{2 m}{(p_{\kappa+1}+2\lambda s_\kappa)}^2 + \frac{mg}{2 \lambda} (e^{2 \lambda q_\kappa} -1),
\end{equation}
where $\lambda, g \in \mathbb{R}$. Calculating the discrete right contact Hamiltonian equations we have the following discrete dynamics:
\begin{align*}
    q_{\kappa+1}&= \frac{1}{m}\left(p_{\kappa+1}+2\lambda s_\kappa\right),\\
    p_\kappa&=\frac{mge^{2\lambda q_\kappa}}{1-\frac{2\lambda}{m}\left(p_{\kappa+1}+2\lambda s_\kappa\right)},\\
    s_{\kappa+1}&=s_\kappa+\frac{p_{\kappa+1}^2}{2m}-\frac{2\lambda^2 s_\kappa}{m}-\frac{mg}{2\lambda}\left(e^{2\lambda q_\kappa}-1\right)
\end{align*}

This dynamic is represented in the following diagram obtained by plotting the points $(q(\kappa),p(\kappa),s(\kappa))$ which have been integrated recursively for different values of $\kappa$.
For $\lambda=-0.01$, $m=1$, $g=10$ and initial condition $(q(0),p(0),s(0))=(100,1,1)$ we obtain the dynamics of the parachute as follows.
\begin{equation}\label{figure}
    \begin{matrix}
\includegraphics[scale=0.4]{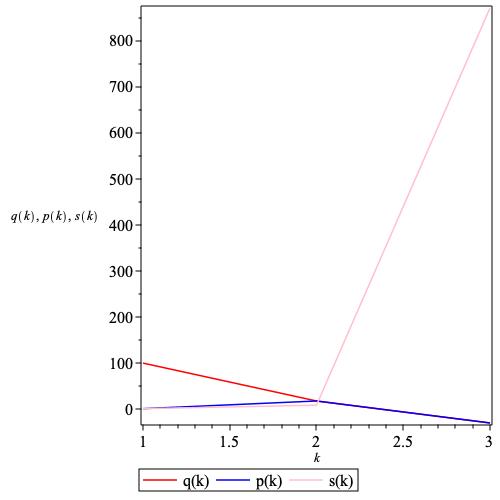} & \includegraphics[scale=0.4]{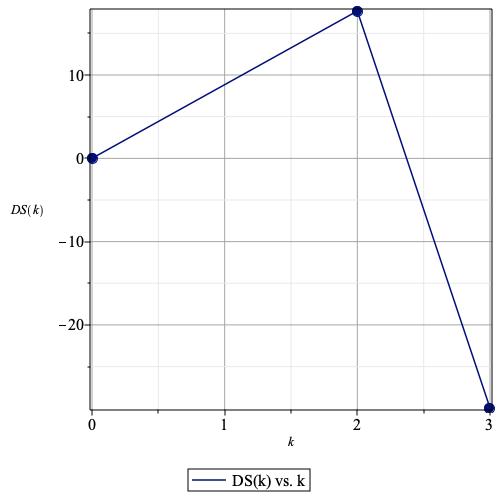}  \\
\text{Figure 1: Dynamics of the parachute} & \text{Figure 2:  Evolution of $DS(k)$}
\end{matrix}
\end{equation}
Here it is easy to depict from the figure \ref{figure} that the position coordinate (red line on the left) descends from the initial position $q(0)=100$ until the parachute reaches the ground. The momentum (blue line on the left)  grows until it crosses the red line. This crossing point represents the instant in which the parachute opens and the momentum decreases until the parachute touches the ground, i.e., the parachute decelerates so the parachuter can touch the ground safely. Nonetheless, as it happens, the velocity is not zero when the parachuter touches the ground.

Now, in order to prove the accuracy of the discrete contact Hamilton--Jacobi equation, we plot the differential of the generating function, i.e., $DS^{\kappa}(q_{\kappa})$, which according to the first equation in \eqref{genfunction}, it should match the dynamic of the blue line in the diagram depicted above. Let us corroborate this fact by plotting $\gamma=DS^\kappa(q_\kappa)$ for different discrete values of $\kappa$. For $\lambda=-0.01$, $m=1$, $g=10$ and initial condition $(q(0),S(0),s(0))=(100,1,1)$ we obtain:

It is easy to check that the line representing $DS(k)$ (blue line on the right figure \eqref{figure}) is very similar to the blue line depicted in blue line on the left. Indeed, both lines should be equal (see that for $k=2$ both reach the maximum momentum $p$) since in the vicinity of the discrete Hamilton--Jacobi equation, $p(k)=DS(k)$.

\section{Conclusions}

In this paper, we have proposed a discrete Hamilton--Jacobi equation (stated in Theorem \ref{theorem-dHJ}) for contact discrete Hamiltonian dynamics. We have presented its geometric foundations in terms of contact discrete flows in Theorem \ref{gHJ-discrete}. We have exhibited the relationship between the continuous geometric HJ Theorem \ref{HJT-Con} for contact Hamiltonian dynamics and Theorem \ref{gHJ-discrete} for discrete contact Hamiltonian dynamics.

We wish to continue in the following directions: 

\begin{itemize} 
\item Contact Hamiltonian dynamics does not preserve the Hamiltonian function. There exists an alternative characterization of Hamiltonian dynamics on contact manifolds that preserves the energy, and it is known as evolution dynamics \cite{Simoes-thermo,Simoes2020b}. We wish to examine the discretization of evolution dynamics and its HJ formulation.    
\item For the extended cotangent bundle $T^*Q\times \mathbb{R}$, the continuous HJ theory for contact Hamiltonian dynamics was presented in \cite{LeonLainz4,LeonSar3,esen2021implicit}. The authors consider the base manifold as the extended configuration space $Q\times \mathbb{R}$. In this work, we consider the base manifold as $Q$. We wish to write a discrete HJ equation on contact manifolds with base manifold $Q\times \mathbb{R}$. 
\item If a Lagrangian is degenerate, then one cannot arrive at explicit Euler-Lagrange equations. In this case, the Legendre transformation is not immediate. Tulczyjew's triple is a geometric formulation that allows us to achieve this in singular cases as well \cite{tul77}. In a discrete framework, Tulczyjew's triple was constructed in \cite{Leok-Tulc}. This determines a proper geometry for implicit discrete Lagrangian and Hamiltonian dynamics \cite{Iglesias-discrete}. For the continuous case, a geometric HJ theory has been recently given in  \cite{EsLeSa18,EsenLeonSar2} in the symplectic framework  and in \cite{esen2021implicit} in the contact framework. We wish to concentrate on generalizing the discrete HJ theories both for symplectic and contact geometry including the implicit case. On the other hand, the Tulczyjew triple for contact geometry has been recently constructed in \cite{esen2021contact}. In the future, we aim at constructing a discrete contact Tulczyjew's triple. 
\end{itemize}

\section*{Acknowledgements}
This work has been partially supported by MINECO MTM 2013-42-870-P and
the ICMAT Severo Ochoa project SEV-2011-0087. Marcin Zajac and Cristina Sardón acknowledge the funding from Universidad Politécnica de Madrid and the department of Applied Mathematics at ETSII for hosting Prof. Marcin Zajac to work alongside Cristina Sardón.

\bibliographystyle{abbrv}
\bibliography{references}

\end{document}